\documentclass[runningheads]{llncs}
\usepackage{graphicx}
\usepackage{caption}
\usepackage{amsmath,amssymb,amsfonts}
\usepackage{algorithmic}
\usepackage{enumitem}
\usepackage{textcomp}
\usepackage{xcolor}
\def\BibTeX{{\rm B\kern-.05em{\sc i\kern-.025em b}\kern-.08em
    T\kern-.1667em\lower.7ex\hbox{E}\kern-.125emX}}
\usepackage{url} 
\usepackage{tikz} 
\usetikzlibrary{positioning,arrows,shapes.misc,shapes.gates.logic.US,calc}
\usetikzlibrary{positioning,shapes,arrows}
\usepackage{varwidth} 
\usepackage{hyperref} 
\usepackage{dcolumn}
\usepackage{booktabs}
\usepackage{tikz}
\newcolumntype{M}[1]{D{.}{.}{1.#1}}
\usepackage{todonotes}
\usepackage{color}
\usepackage{listings}
\usepackage{amsmath, amssymb}
\usepackage[boxed]{algorithm2e}
\SetKwProg{Fn}{def}{\string:}{}
\SetKwFunction{FRecurs}{RecursiveProbability}%
\SetKw{KwTo}{in}\SetKwFor{For}{for}{\string:}{}%
\SetKwIF{If}{ElseIf}{Else}{if}{:}{elif}{else:}{}%
\SetKwFor{While}{while}{:}{fintq}%
\SetAlgoNoEnd\SetAlgoNoLine%
\DontPrintSemicolon
\makeatletter
\newcommand{\removelatexerror}{\let\@latex@error\@gobble}
\makeatother
\begin{document}
\title{Stochastic Simulation Techniques for Inference and Sensitivity Analysis of Bayesian Attack Graphs}

\titlerunning{Stochastic Simulation for Bayesian Attack Graphs}
%
\author{Isaac Matthews\inst{1,2} \and
Sadegh Soudjani\inst{1} \and
Aad van Moorsel\inst{1}}
\authorrunning{I. Matthews et al.}
%
\institute{School of Computing, Newcastle University, United Kingdom \and
\email{I.J.Matthews2@newcastle.ac.uk}\\
}
\maketitle              
\begin{abstract}
A vulnerability scan combined with information about a computer network can be used to create an attack graph, a model of how the elements of a network could be used in an attack to reach specific states or goals in the network. These graphs can be understood probabilistically by turning them into Bayesian attack graphs, making it possible to quantitatively analyse the security of large networks. In the event of an attack, probabilities on the graph change depending on the evidence discovered (e.g., by an intrusion detection system or knowledge of a host's activity). Since such scenarios are difficult to solve through direct computation, we discuss and compare three stochastic simulation techniques for updating the probabilities dynamically based on the evidence and compare their speed and accuracy.  From our experiments we conclude that likelihood weighting is most efficient for most uses.  We also consider sensitivity analysis of BAGs, to identify the most critical nodes for protection of the network and solve the uncertainty problem in the assignment of priors to nodes.  Since sensitivity analysis can easily become computationally expensive, we present and demonstrate an efficient sensitivity analysis approach that exploits a quantitative relation with stochastic inference.   

\end{abstract}
%
%
%

%
%

\section{Introduction}
Attack graphs are models of how vulnerabilities can be exploited to attack a network. They are directed graphs that demonstrate how multiple vulnerabilities and system configurations can be leveraged during a single attack in order to reach states in the network that were previously inaccessible to the attacker. An example of such a state would be root privilege on a database that contains sensitive information.
Attack graphs can be generated by performing a vulnerability scan of a network using a tool like OpenVAS~\cite{openvas} or Nessus~\cite{nessus}, and then processing the results with an attack graph generator. There are many generators available (NetSPA \cite{Ingols2009}, TVA \cite{jajodia2009topological}, MulVAL \cite{Ou2005:MLN:1251398.1251406}) but for this paper we use MulVAL as it is open source and is used by the majority of the literature on attack graphs. 

Attack graphs can be combined with Bayesian networks to allow for a probabilistic analysis of the security of a network \cite{Aguessy2016,Huangfu2017,Munoz-Gonzalez2017,Ramaki2015,sembiring2015network,Dantu2004,Doynikova2017,homer2013aggregating}. This combination is known as  a Bayesian Attack Graph (BAG). These BAGs are generated by incorporating likelihood information for vulnerabilities into the original attack graph. One method of constructing BAGs is to acquire information on each vulnerability that is present in the graph from a vulnerability repository like the National Vulnerability Database.  The Common Vulnerability Scoring System (CVSS) \cite{firstorg} vector that is in the database contains information, e.g., on the attack complexity for exploiting a vulnerability and the availability of an exploit \cite{firstorg}, which can be used in various ways to estimate the likelihood that the vulnerability will be exploited \cite{Frigault2017,Cheng2017,Keramati2014}. 

BAGs are particularly promising as a dynamic risk assessment tool where an administrator models new security controls and their effects on a network. A network's most likely attack paths and most vulnerable hosts can be dynamically analysed, and this can be updated dependent on information from an intrusion detection system \cite{Poolsappasit2012}.


Well-defined (that is, acyclic) BAGs can be solved using computational techniques that are well-known from the theory of Bayesian Networks \cite{Nielsen2009bayesian}. In recent work systematic approaches have also been proposed for BAGs that have loops and cycles, e.g., \cite{matthews2020cyclic}. However, direct computational approaches become prohibitively slow if the number of nodes in the BAG is large, and can have large space requirements due to an increase in the size of the cliques in the graphs and their probability tables. Therefore, it becomes important to consider stochastic simulation (Monte-Carlo) techniques.   

In this paper we focus on performing inference and sensitivity analysis on BAGs using stochastic simulation. We do this for dynamic scenarios that do not lend themselves for exact computation, namely scenarios that include observed evidence in the BAGs. We discuss how any evidence or alterations to the network can be included in the BAG analysis. That is, we create a dynamic model of the security of the network that can be used to deduce an attacker's most likely next move and their route thus far, as well as quantitatively evaluate and compare the effectiveness of different security controls and changes to the network. 

While inference for BAGs using stochastic simulation has been performed by others to investigate potential uses \cite{Noel2010measuring,baiardi2013assessing}, there has to date not been a comparison of different techniques' performances on BAGs. In this paper we employ three stochastic simulations techniques, probabilistic logic sampling (PLS), likelihood weighting (LW), and backward simulation (BS). We evaluate the performance of these techniques in their speed and accuracy as well as how they perform with different quantities of evidence to be included in the graph and different sizes of graph.

The primary outcome of our work is a recommendation of the most efficient simulation technique to use for inference in attack graphs. The recommendation is to use likelihood weighting, which performs well for both low and high evidence scenarios. 
Moreover, we establish a quantitative relation between stochastic inference and sensitivity analysis of BAGs.   
We discuss how the methods for including evidence in the graphs can also be used to measure the graphs sensitivity to each vulnerability in the network, and develop a fast approach to calculate these sensitivities without requiring many simulations or any analysis of distributions.

The rest of this paper is organised as follows. Section~\ref{sec:motivation} describes the formalism for BAGs that is used throughout the paper and introduces the running example along with the motivation for the work. Section~\ref{sec:samplingtech} introduces and discusses the three sampling techniques that are implemented and their accuracy. Section~\ref{sec:comparison} then evaluates the performances of these techniques with regard to accuracy, amount of evidence and size of graph. Section~\ref{sec:sens} introduces our measure of sensitivity and its importance. Finally Section~\ref{sec:related} compares this work with the current literature available and Section~\ref{sec:conclusion} presents our conclusions.


\section{Bayesian Attack Graphs}
\label{sec:motivation}

We consider a small enterprise network as a standard example used in the literature \cite{Ou2011,homer2013aggregating,matthews2020cyclic}
to motivate the use of Bayesian attack graphs (BAGs) and demonstrate the sampling techniques discussed in this paper.

\subsection{Motivating Example}

\begin{figure}
    \centering
    \includegraphics[width=0.6\linewidth]{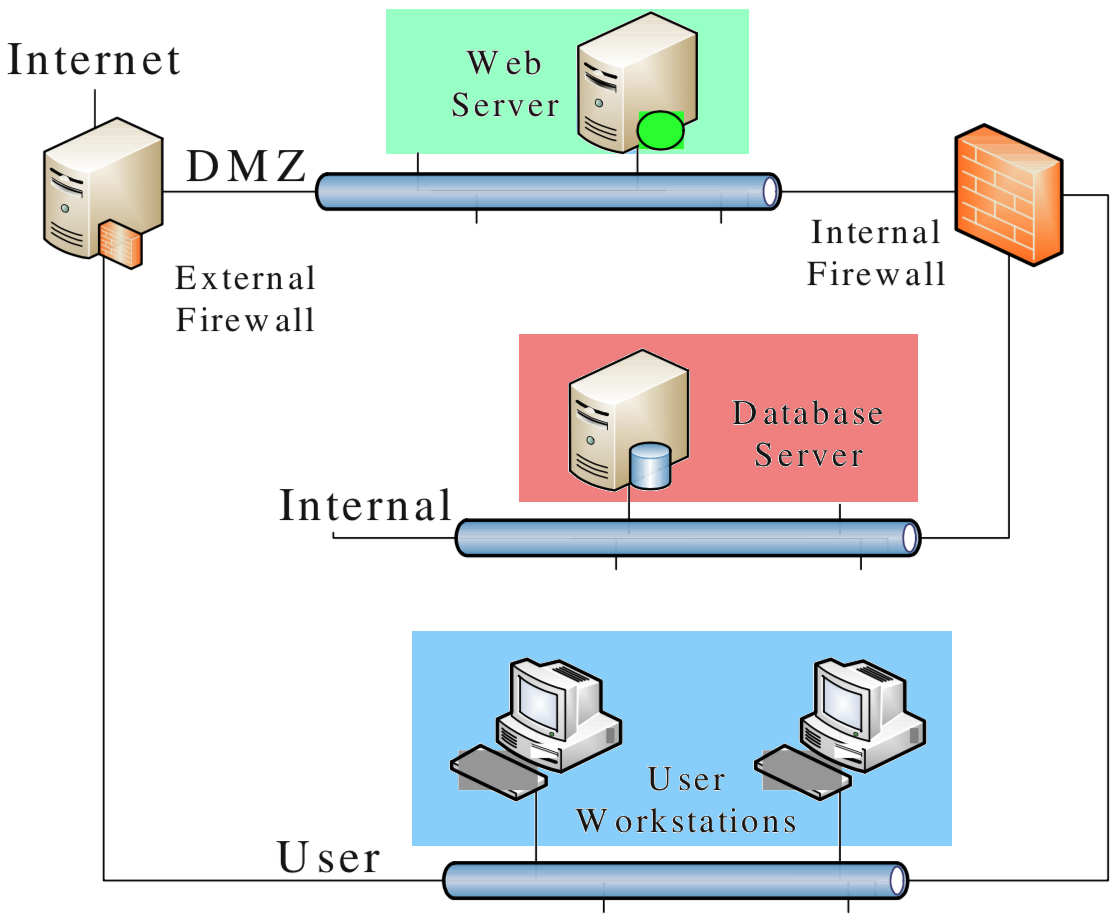}
    \caption{Example small enterprise network architecture.}
    \label{fig:netarch}
\end{figure}

The architecture of the small enterprise network can be seen in Figure~\ref{fig:netarch}. In this scenario, the network administrator wants to protect the database server on the internal network from being accessed and the data being exfiltrated. The internal network with the database server can be accessed via the internal firewall one of two ways, either from the web server or from one of the workstations (grouped together and treated as one host for simplicity). Both of these routes require access to the demilitarized zone subnet (DMZ), which can be accessed by the internet through the external firewall.

We can run a vulnerability scan on this network using tools like OpenVAS~\cite{openvas} or Nessus~\cite{nessus} to create an attack graph. In this scenario, we assume that the scan discovers a vulnerability on each of the hosts. On the database server there is a MySQL vulnerability, the web server has a vulnerability in Apache, and the workstations have an Internet Explorer vulnerability. A full description of the vulnerabilities and the resulting attack graph can be found in Appendix~\ref{app:eg}. With this vulnerability scan, we can generate an attack graph, for instance using the tools provided in MulVAL~\cite{Ou2005:MLN:1251398.1251406}. The attack graph represents how the vulnerabilities can be used in conjunction with one another to reach a high enough privilege on the database server to access the data.

\subsection{Bayesian Attack Graphs}

\begin{figure}
    \centering
	\includegraphics[trim=40 38 40 38, clip, width=\linewidth]{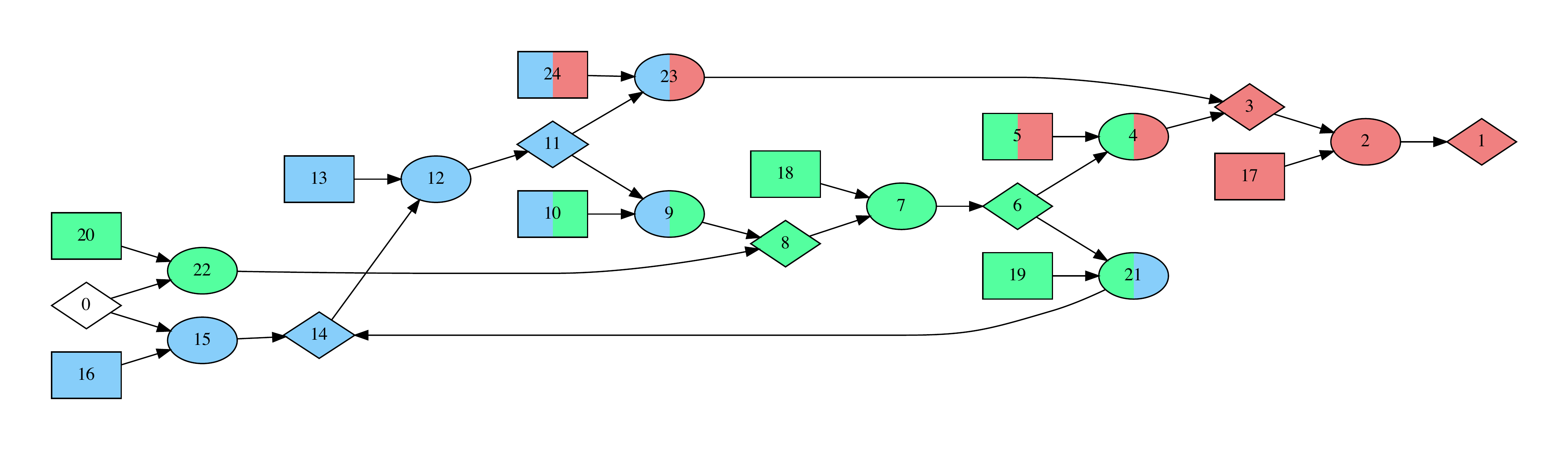}
	\caption{The BAG of the small enterprise network presented in Figure~\ref{fig:netarch}.}
	\label{fig:egbag}
\end{figure}

Figure~\ref{fig:egbag} shows the resulting BAG from running a scan on the network presented in Figure~\ref{fig:netarch}. For clarity, the nodes have been coloured to show which of the hosts in the architecture they correspond to. Nodes that have multiple colours are related to a transition between two hosts (from the colour on the left to the colour on the right). Node 0 in white colour is not related to any host on the network but represents an attack from the internet. Node 1 is deemed the goal node as it represents the primary state that the attacker must not reach, or else the data in the database becomes accessible to them. In general the goal node or nodes is the collection of nodes that allow the attacker to achieve something that the network administrator is trying to prevent. For a full description of this graph, along with all node labels, see Appendix \ref{app:eg}. 

We provide the formal definition of BAGs next.

\begin{definition}
\label{def:BAG2}
A Bayesian attack graph is defined as a directed graph $\mathcal G = (\mathcal V,\mathcal E),$ where $\mathcal V$ is the set of nodes and $\mathcal E\subset \mathcal V\times \mathcal V$ is the set of edges. Nodes in $\mathcal V$ are connected by edges from $\mathcal E$. We denote the edge connecting $v_i,v_j\in\mathcal V$ by $e_{ij}=(v_i,v_j)$. The set of nodes is comprised of three types of nodes, $\mathcal V = V_l \cup V_a \cup V_o$, representing LEAF, AND and OR nodes, as explained below.
\end{definition}
For any $v,v'\in\mathcal V$ and $(v,v')\in\mathcal E$, $v$ is called the parent node of $v'$ and $v'$ is called a child node of $v$. Similarly, we have the set of parent nodes $pa(v') := \{v\in\mathcal V\,|\, (v,v')\in\mathcal E\}$. 

The three types of nodes are as follows, first introduced in informal terms, then formally within the context of Bayesian attack graphs:
\begin{itemize}[noitemsep,topsep=4pt]
\item
$V_l$ is the set of LEAF nodes in the graph, nodes that have no parent. They represent network configurations, the existence of vulnerablities or running services, and different conditions in the network, for example network connection information in the form of HACLs (host access control lists).
\item
$V_a$ are the AND nodes, which have requisite conditions {\em all} of which must be satisfied in order to be accessed (an AND nodes parents have a conjunctive relationship). An example of an AND node would be the remote exploitation of a vulnerability, given that the vulnerability exists and the attacker has access to a host that is allowed access to the machine the vulnerability is located on. An AND node could also represent a movement between hosts if there are the fulfilled requirements of a configuration node for access between the machines and the attacker has access to one of them already.
\item
$V_o$ are the OR nodes, which have requisite conditions of which {\em at least one} must be satisfied in order to be accessed (an OR node's parents have a disjunctive relationship). These nodes represent micro-states within the network that encode information about an attacker's location and privilege in the system. For example, if a machine has several vulnerabilities that could be exploited to achieve privilege escalation on that machine then exploiting any of those vulnerabilities would grant the attacker access to the state of higher privilege on the host. The overall macro-state of the attacker, being their privilege on each host of the system and their access to and between each of the hosts, would be an enumeration of all the OR nodes that had been accessed.
\end{itemize}

As indicated in Figure~\ref{fig:egbag}, the LEAF nodes are drawn as rectangles on the graph, the AND nodes are ellipses and the OR nodes are diamonds.


Suppose a local probability function $p:\mathcal V\rightarrow [0,1]$ is given.
Any BAG $\mathcal G = (\mathcal V,\mathcal E)$ as in Definition~\ref{def:BAG2} with local probability function $p:\mathcal V\rightarrow [0,1]$ can be translated into a Bayesian Network, which is denoted by the tuple $\mathfrak B = (\mathcal V,\mathcal E,\mathcal T)$. Let us consider the set of nodes $\mathcal V = \{1,2,\ldots,n\}$ are associate a Boolean random variable $X_k$ to each node $k\in\mathcal V$. The set $\mathcal T$ is a collection of probability tables that are constructed as follows. For all $k\in V_l$,
\begin{equation}
\label{eq:BN_leaf}
\text{Prob}(X_k=1) = p(k) \quad \text{ and }\quad  \text{Prob}(X_k=0) = 1-p(k).
\end{equation}
For $k\in V_a$, let $pa(k)=\mathbf{1}$ indicate that all the parent variables of the node $k$ are equal to one. Then,
\begin{equation}
\label{eq:BN_and}
\begin{cases}
\text{Prob}(X_k=1|pa(k)=\mathbf{1}) = p(k),\\
\text{Prob}(X_k=1|pa(k)\neq\mathbf{1}) = 0,\\
\text{Prob}(X_k=0|pa(k)=\mathbf{1}) = 1-p(k),\\
\text{Prob}(X_k=0|pa(k)\neq\mathbf{1}) = 1.
\end{cases}
\end{equation}
For $k\in V_o$, let $pa(k)=\mathbf{0}$ indicate that all the parent variables of the node $k$ are equal to zero. Then,
\begin{equation}
\label{eq:BN_or}
\begin{cases}
\text{Prob}(X_k=1|pa(k)=\mathbf{0}) = 0,\\
\text{Prob}(X_k=1|pa(k)\neq\mathbf{0}) = p(k),
\end{cases}
\end{equation}
and $\text{Prob}(X_k=0|pa(k))$ is the complement of above probabilities.



\begin{remark}
Note that here we use the AND/OR formalism for BAGs \cite{Ou2006:SAA:1180405.1180446} having three types of nodes, but another common formalism is the plain BAG \cite{Swiler2001} that has only one type of nodes. We use AND/OR BAGs in order to move probabilities to the edge of the graph to allow deterministic analysis for all other nodes as described below. In our previous work \cite{matthews2020cyclic}, we have demonstrated the relation between the two formalisms and explained how to transform one form to the other and as such the current paper is relevant to both types of BAGs.
\end{remark}

\subsection{Problem Statement}
\label{problem}



\begin{problem}[Access Probabilities]
\label{prob:access}
Consider an attack graph $(\mathcal{V}, \mathcal{E})$ with local probability function $p:\mathcal V\rightarrow [0,1]$. Compute $\text{Prob}(X_k=1)$ for any $k\in \mathcal V$. This quantity is called the \emph{access probability} of the node $k$ and is simply denoted by $P(X_k)$. It will give the likelihood that an attacker will reach node $k$ in an attack and will depend on the local probability function $p$ and the structure of the attack graph.
\end{problem}

\begin{problem}[Inference]
\label{prob:inference}
Suppose some evidence of an attack is known in the form of $\mathcal{Z} = \textbf{z}$, where $\mathcal{Z} \subset \mathcal{V}$ is the set of random variables associated to the nodes for which we have the respective evidence values $\textbf{z}$. Compute the likelihood that the attacker gain access to node $k\in\mathcal V$ given such an evidence: $P(X_k|\mathcal{Z}=\textbf{z})$.
\end{problem}

\begin{problem}[Sensitivity Analysis]
\label{prob:sense}
The local probability function $p:\mathcal V\rightarrow [0,1]$ is often estimated based on prior knowledge or data on the network. 
Compute the sensitivity of access probabilities $P(X_k)$ and conditional probabilities $P(X_k |\mathcal{Z}=\textbf{z})$ to the values $p(v)$ for any $v\in\mathcal V$. If $p(v)$ has a distribution, compute an interval for these quantities with a confidence bound.
\end{problem}

We provide stochastic simulation techniques to answer Problems~\ref{prob:access} and \ref{prob:inference} in Section~\ref{sec:samplingtech}, and present a novel solution to Problem~\ref{prob:sense} in Section~\ref{sec:sens}.

\section{Sampling Techniques}
\label{sec:samplingtech}

\subsection{Graph Decomposition}

In order to simplify the process of simulation for attack graphs, we can move all stochastic behaviour to LEAF nodes and in doing so make the rest of the graph purely deterministic. This can be achieved by enlarging the attack graph and moving local probabilities of non-LEAF nodes onto a new LEAF node with the same local probability. This process can be seen in Figures~\ref{fig:precon}-\ref{fig:postcon}, where the node C will have the same behaviour in both figures. This is only required when non-LEAF nodes have local probabilities which is not the case for graphs we create but is found in the literature. Full demonstration of the equivalence is in our previous work \cite{matthews2020cyclic}.
\begin{figure}
\centering    
\begin{minipage}{.5\textwidth}
    \begin{tikzpicture}[node distance=1cm and 0cm]
    \node[draw, ellipse, minimum width=1.1cm, align=center] (A) {A};
    \node[draw, ellipse, minimum width=1.1cm, align=center, right= 0.75cm of A] (B) {B};
    \path (A) -- (B) coordinate[midway] (mid);
    \node[draw, diamond, aspect=2, minimum size=1cm, below = of mid, align=center] (C) {C};
    \path (A) edge[-latex] (C);
    \path (B) edge[-latex] (C);
    \node[below right=0.02cm of B]
        {
        \begin{tabular}{cccc}
        \toprule
        && \multicolumn{2}{c}{C} \\
        A & B & \multicolumn{1}{c}{F} & \multicolumn{1}{c}{T} \\
        \cmidrule(r){1-2}\cmidrule(l){3-4}
        F & F & 1 & 0 \\
        F & T & 1-$p(v)$ & $p(v)$ \\
        T & F & 1-$p(v)$ & $p(v)$ \\
        T & T & 1-$p(v)$ & $p(v)$ \\
        \bottomrule
        \end{tabular}
        };
    \end{tikzpicture}
    \caption{A small attack graph, with local probability in OR node C.}
    \label{fig:precon}
\end{minipage}%
\begin{minipage}{.5\textwidth}
    \begin{tikzpicture}[node distance=1cm and 0cm]
    \node[draw, ellipse, minimum width=1.1cm, align=center] (A) {A};
    \node[draw, ellipse, minimum width=1.1cm, align=center, right= 0.75cm of A] (B) {B};
    \path (A) -- (B) coordinate[midway] (mid);
    \node[draw, diamond, aspect=2, minimum size=1cm, below = of mid, align=center] (C1) {C$^\prime$};
    \node[draw, minimum size=0.8cm, align=center, right=0.75cm of C1] (v) {$v'$};
    \path (C1) -- (v) coordinate[midway] (mid1);
    \node[draw, ellipse, minimum width=1.1cm, align=center, below= 1cm of mid1] (C) {C};
    \path (A) edge[-latex] (C1);
    \path (B) edge[-latex] (C1);
    \path (v) edge[-latex] (C);
    \path (C1) edge[-latex] (C);
    \node[right=0.15cm of v]
        {
        \begin{tabular}{cc}
        \toprule
        \multicolumn{2}{c}{$v'$} \\
        \multicolumn{1}{c}{F} & \multicolumn{1}{c}{T} \\
        \cmidrule(l){1-2}
        1-$p(v)$ & $p(v)$ \\
        \bottomrule
        \end{tabular}
        };
    \node[below left=0.15cm of C1]
        {
        \begin{tabular}{cccc}
        \toprule
        && \multicolumn{2}{c}{C$^\prime$} \\
        A & B & \multicolumn{1}{c}{F} & \multicolumn{1}{c}{T} \\
        \cmidrule(r){1-2}\cmidrule(l){3-4}
        F & F & 1 & 0 \\
        F & T & 0 & 1 \\
        T & F & 0 & 1 \\
        T & T & 0 & 1 \\
        \bottomrule
        \end{tabular}
        };
    \node[below=0.15cm of C]
        {
        \begin{tabular}{cccc}
        \toprule
        && \multicolumn{2}{c}{C} \\
        C$'$ & $v'$ & \multicolumn{1}{c}{F} & \multicolumn{1}{c}{T} \\
        \cmidrule(r){1-2}\cmidrule(l){3-4}
        F & F & 1 & 0 \\
        F & T & 1 & 0 \\
        T & F & 1 & 0 \\
        T & T & 0 & 1 \\
        \bottomrule
        \end{tabular}
        };
    \end{tikzpicture}
    \caption{The equivalent graph of Figure~\ref{fig:precon} with probabilities only on leaves.}
    \label{fig:postcon}
\end{minipage}
\end{figure}

\subsection{Generating Samples}

Using this formalism of BAGs a single attack can be modelled as the array of LEAF nodes being allocated values corresponding to `achieving' something in an attack such that if the node is given the value 1 then the exploit has worked or a condition has been met, and if the value is 0 then an exploit or condition has failed. With these values the reach of the attack can be calculated, as the internal (non-LEAF nodes) in the graph are all deterministically dependent on the LEAF nodes. In this way, with a specific distribution of LEAF nodes a state on the graph is either accessible or inaccessible. With the example of Figure \ref{fig:egbag}, a single attack configuration would equate to all the rectangular LEAF nodes being set as $y$ or $n$, determining the states of the rest of the nodes in the graph. The attackers ability to reach an important state, like the ability to execute code on the database server at node 1, becomes either $y$ or $n$. 


We prepare the graph by assigning the LEAF nodes a series of prior distributions based upon factors like the ease of exploitation of a vulnerability. We then sample from these to create a single attack simulation with the LEAF nodes being 1 or 0 according to a random sample of their distribution, and all other nodes being assigned values deterministically from their tables.

\subsection{Probabilistic Logic Sampling}
For our attack graphs, probabilistic logic sampling (PLS) is performed by first sampling a configuration of LEAF nodes. A random number is generated between 0 and 1, if the number generated is less than the prior probability assigned to the LEAF node then the node is assigned a 1 (or $y$), if it is greater then the node is set at 0 ($n$). This is repeated for all LEAF nodes to create the configuration. When the configuration has been generated it can be used to prescribe states to the rest of the nodes in the graph. These states are then recorded as an array of 1s and 0s. This process is then repeated until $N$ configurations have been generated and evaluated. The recorded arrays can be used to estimate the probability distributions of the nodes in the graph, with $\frac{N(X=1)}{N}$ being the estimated probability that an attacker will gain access to a particular node:
\begin{equation}
\label{eq:probs}
    P(X) \approx \left(\frac{N(X=1)}{N},\frac{N(X=0)}{N}\right)
\end{equation}
The simplest way to include evidence with this technique is by discarding any samples that do not conform to the evidence provided. As such one is left with a subset of the original $N$ simulations and can calculate the new probabilities in a similar way to equation \eqref{eq:probs}.

In order to estimate the probability distribution of the $k^{th}$ variable with regard to the new evidence, $P(X_k|\mathcal{Z}=z)$, using $N$ samples with PLS we use algorithm \ref{alg:pls} modified from \cite{Nielsen2009bayesian}. Here $\mathcal{Z}$ is the variables or nodes that we have evidence for and $\textbf{z}$ is the evidence that has been provided for these nodes. This would be in the form of a list of nodes that we know have been accessed created by an intrusion detection system, or a list of nodes that we are modelling as not accessed if we are comparing different security controls and their affect on the network.

\begin{algorithm}
\removelatexerror
	\SetAlgoLined
	\KwIn{attack graph $(\mathcal{V}, \mathcal{E})$, local probability function $p:\mathcal V\rightarrow [0,1]$, number of simulations $N$}
	\KwOut{Conditional likelihood $P(X_k|\mathcal{Z}=\textbf{z})$ given evidence $\mathcal{Z}=\textbf{z}$}
	
	1. Let $\mathcal V = \{1,2,\ldots,n\}$ and $(X_1,...,X_n)$ be the associated Boolean random variables.
	
	2. Initialise $N(X_k = x_k) = 0$ for all $x_k\in\{0,1\}$ and all $k\in\mathcal V$.
	
    3. \For{$j=1$ to $N$}{
	a) \For{$i=1$ to $n$}{
	Sample a state $x_i$ for $X_i$ using $P(X_i|pa(X_i) = \pi)$, where $\pi$ is the configuration sampled for $pa(X_i)$
	}
	b) If $\textbf{x} = (x_1,...,x_n)$ is consistent with \textbf{z}, then $N(X_k = x_k) := N(X_k = x_k) + 1$, where $x_k$ is the sampled state for $X_k$
	}
	\Return{ $\text{Prob}(X_k = x_k|\textbf{z}) \approx \dfrac{N(X_k = x_k)}{N(X_k=0) + N(X_k=1)}$}
\caption{Performing PLS to approximate a distribution given some evidence.}
\label{alg:pls}
\end{algorithm}

\subsection{Likelihood Weighting}
Likelihood weighting (LW) is a method to deal with the problems of PLS for dealing with evidence, namely the inefficiency of generating samples that will be discarded if they conflict with evidence. Instead, for LW, only non-evidence variables are sampled from and as such no simulations are discarded. However this approach causes sampled variables to ignore evidence that is not present in their ancestors, and so an extra weighting has to be introduced. This weighting is equivalent to the probability a certain state will arise given the evidence provided.

Essentially we want to sample from the following distribution, 
\begin{equation}
\begin{aligned}
\label{eq:lwaim}
    P(\textbf{V},\textbf{z}) = &\prod_{X\in V\setminus\mathcal{Z}} P(X|pa(X)`,pa(X)`` = \textbf{z})  \\
    &\times \prod_{X\in \mathcal{Z}} P(X=e|pa(X)`,pa(X)`` = \textbf{z}) 
\end{aligned}
\end{equation}
where $pa(X)``$ are parent nodes that have evidence, and $pa(X)`$ do not. By fixing the evidence variables then taking the sample we instead are using
\begin{equation}
\label{eq:lwsample}
    Distribution = \prod_{X\in V\setminus\mathcal{Z}} P(X|pa(X)`,pa(X)`` = \textbf{z}) 
\end{equation}
So to rectify this we weigh each sample taken using
\begin{equation}
\label{eq:lwweight}
     w(\textbf{x},\textbf{z}) = \prod_{Z\in \mathcal{Z}} P(Z=z|pa(X)=\pi) 
\end{equation}
where $\pi$ is the configuration of the parents specified by \textbf{x} and \textbf{z}.
In order to estimate $P(X_k|\mathcal{Z}=z)$ using $N$ samples we use algorithm \ref{alg:lw} as defined in \cite{Nielsen2009bayesian}.

\begin{algorithm}
\removelatexerror
	\SetAlgoLined
	\KwIn{$P(X_k)$}
	\KwOut{$P(X_k|\mathcal{Z}=z)$}
	
	1. Let $(X_1,...,X_n)$ be all nodes present in the graph.
	
    2. \For{$j=1$ to $N$}{
	a) $w:=1$
	
	b) \For{$i=1$ to $n$}{
	- Let $\textbf{x`}$ be the configuration of $(X_1,...,X_{i-1})$ specified by \textbf{e} and previous samples
	
	- \If{$X_i \not\in \mathcal{Z}$}{
	Sample a state $x_i$ for $X_i$ using $P(X_i|pa(X_i) = \pi)$, where $pa(X_i) = \pi$ is consistent with $\textbf{x`}$
	}
	\Else{
	$w:=w\cdot P(X_i = z_i|pa(X_i) = \pi)$, where $pa(X_i) = \pi$ is consistent with $\textbf{x`}$
	}
	}
	
	c) $N(X_k = x_k) := N(X_k = x_k) + w$, where $x_k$ is the sampled state for $X_k$
	}
	
	\Return{ $P(X_k = x_k|\textbf{z}) \approx \frac{N(X_k = x_k)}{\sum_{x\in sp(X_k)} N(X_k=x)}$}
	\caption{Performing likelihood weighting to approximate a distribution given some evidence}
\label{alg:lw}
\end{algorithm}

This is an improvement on PLS as it removes the inefficiency of discarding evidence, instead requiring the calculation of a weight for each simulation. A large number of samples may still be required, however, if the evidence provided is unlikely and therefore the difference between equations \ref{eq:lwaim} and \ref{eq:lwsample} is large. This would mean the weighting would in general be very small and as such reaching an amount of error that is not too large may take some time.

\subsection{Backward Simulation}
The final technique is based on the Backward Simulation (BS) method devised by Fung and Del Favero \cite{Fung1994backward}. The primary difference between this and other techniques is that simulation runs originate at the known evidence and the simulation is run backwards. Once this process has terminated the remaining nodes are forward sampled in the standard way. The reason for this is to rectify the slow convergence caused by unlikely evidence.

The backward sampling procedure is performed by taking a sample from the distribution
\begin{equation}
P_s(Pa(X_i)) = \frac{P(X_i|pa^u(X_i)pa^*(X_i))}{Norm(i)},i\in N_b
\end{equation}
where $pa^u(X_i)$ are the uninstantiated parents of $X_i$ and $pa^*(X_i)$ are the instantiated parents. The normalizing constant is calculated as 
\begin{equation}
    Norm(i) = \sum_{y\in XP(pa^u(X_i)}P(X_i|y, X_{pa^*(X_i)})
\end{equation}
. Once all backwards sampling nodes have been sampled the forwards sweep samples all the remaining nodes as
\begin{equation}
    P_s(X_i)=P(X_i|pa(X_i)), i\in N_f
\end{equation}
Once a sample has been taken for each node, the weight for the simulation can be computed as the product of the normalisation constants used along with the likelihood of nodes that were set by backwards sampling but were not sampled from themselves
\begin{equation}
    Z(x)=\prod_{i\in N \setminus N_S}P(X_i|Pa(X_i))\prod_{j\in N_b} Norm(j)
\end{equation}
BS, as a form of likelihood weighting, is designed to cope better with very low-likelihood evidence. A large part of the computational cost of the algorithm comes from the calculation of the normalisation constants, which grows exponentially with the number of predecessor nodes. We would expect this technique to perform similarly to likelihood weighting for few evidence nodes but be an improvement when there are many nodes, as is demonstrated in the paper presenting the technique \cite{Fung1994backward}. However the structure of the graph is of great importance and as such it is difficult to know beforehand which of the techniques will perform better for the application of BAGs.

\subsection{Confidence bounds}
Since all these techniques are sampling from the same distribution once the corrective factors are applied, the standard error can be calculated in a similar way for each. As each trial is random and independent from the last, using the central limit theorem it can be shown that
\begin{equation}
    \sigma_{p(x)} = \sqrt{\frac{P\{x\}(1-P\{x\})}{N}}
\end{equation}

\begin{figure}[h!]
\centering
\includegraphics[width=0.7\textwidth]{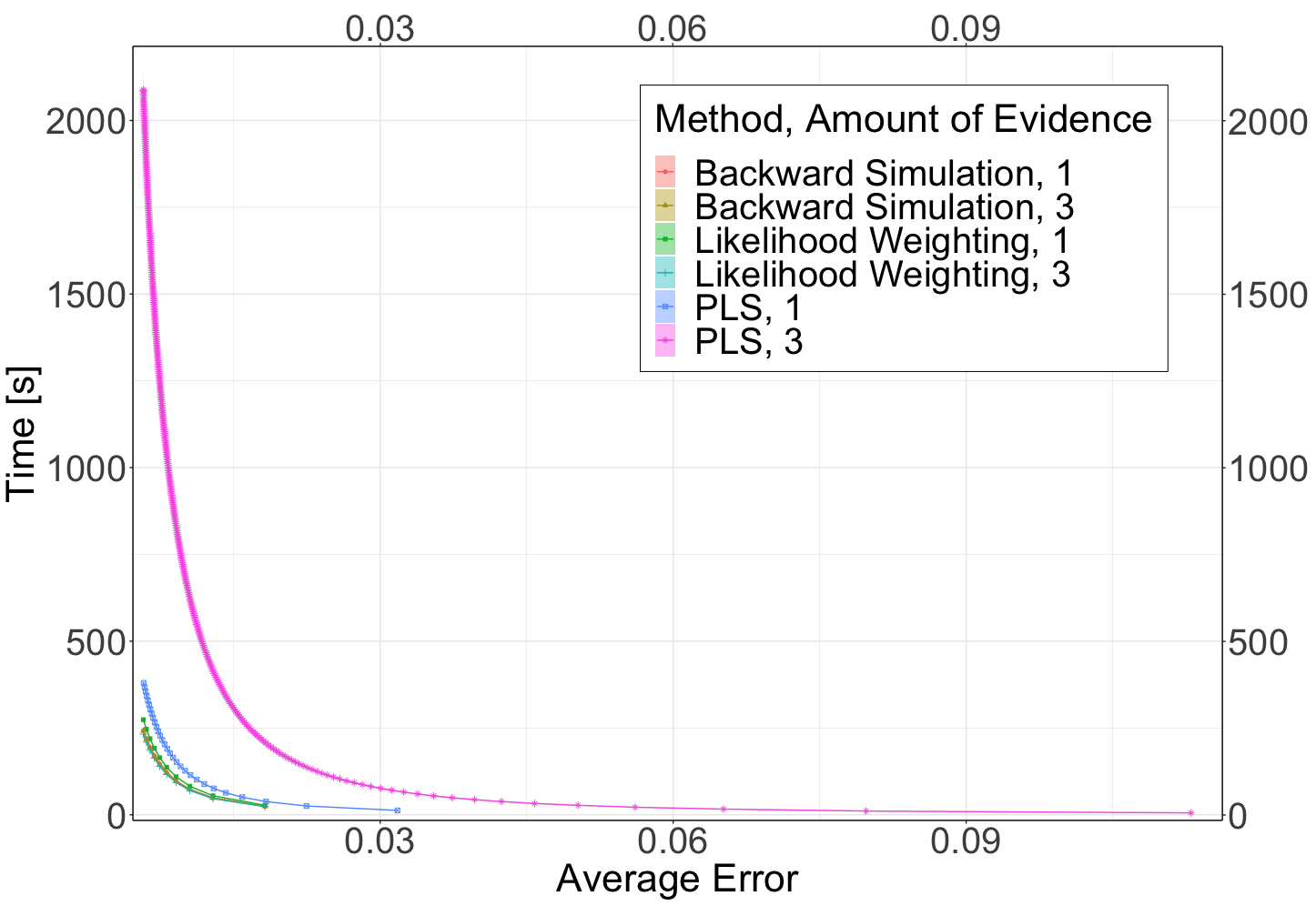}
\caption{Time against average node error for all techniques for one and three evidence nodes.} \label{fig_plstime}
\end{figure}

\section{Comparison}
\label{sec:comparison}
For the comparison of these techniques, each is first run on a 200 node attack graph from a small enterprise network with varying amounts of evidence. Figure \ref{fig_plstime} shows the increase in time (in wall clock seconds) required for improving the accuracy of results for situations when one and three evidence nodes have been included (the average time over thirty runs has been plotted; the error bars are too small to be drawn for this graph). As can be seen even with just one piece of evidence PLS performed poorly compared to the other methods, with three evidence nodes taking considerable amounts of time and runs with more than three evidence nodes timing out. The other two methods are run with five and ten evidence nodes provided, and the results for this can be seen in Figure \ref{fig_nonplstime}, again with the average result over thirty runs plotted. The minimum and maximum values are shown by the transparent ribbon. While these results are close, interestingly LW does outperform BS at higher quantities of evidence.

\begin{figure}[h!]
\centering
\includegraphics[width=0.8\textwidth]{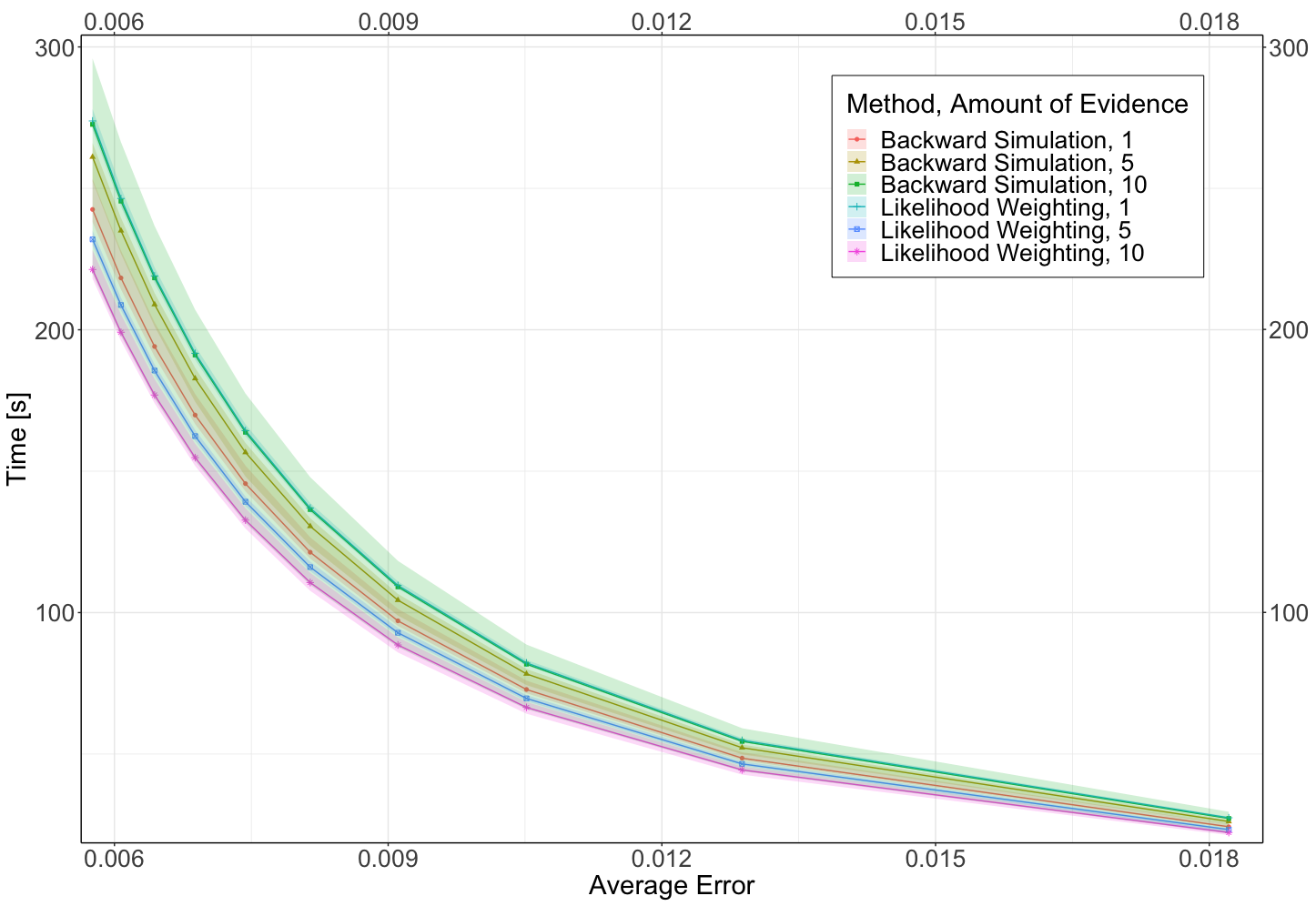}
\caption{Time against average node error for BS and LW for different numbers of evidence nodes.} \label{fig_nonplstime}
\end{figure}


Figure \ref{fig_conv} shows the convergence of each technique on a probability for one of the goals in the network, with the ribbon showing the error of the estimate. LW and BS converge equally quickly with three pieces of evidence but BS does converge faster when only one evidence is used.
\begin{figure}
\centering
\begin{minipage}{.48\textwidth}
\captionsetup{width=.9\linewidth}
  \includegraphics[width=0.99\textwidth]{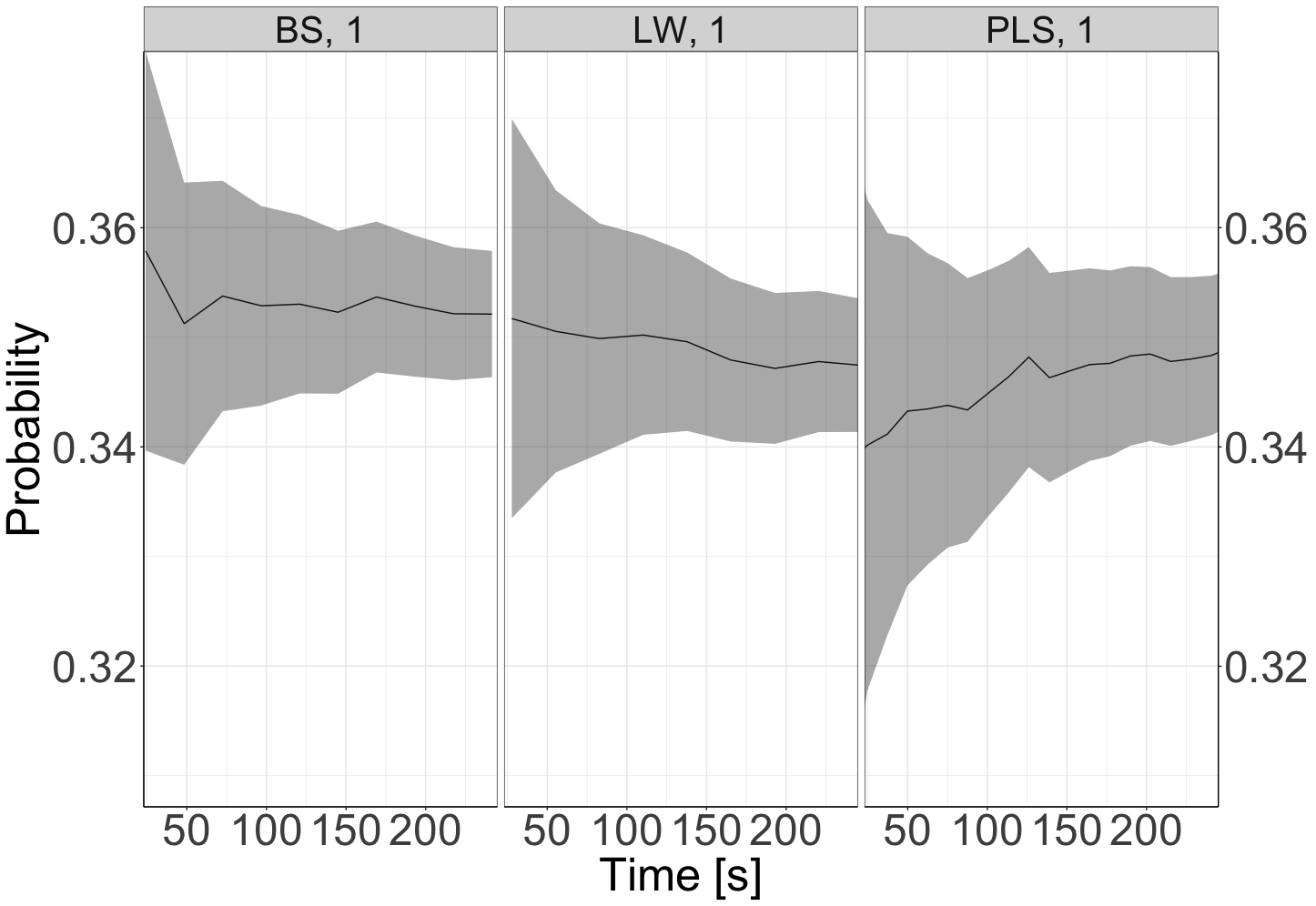}
\end{minipage}%
\begin{minipage}{.48\textwidth}
\captionsetup{width=.9\linewidth}
  \includegraphics[width=0.99\textwidth]{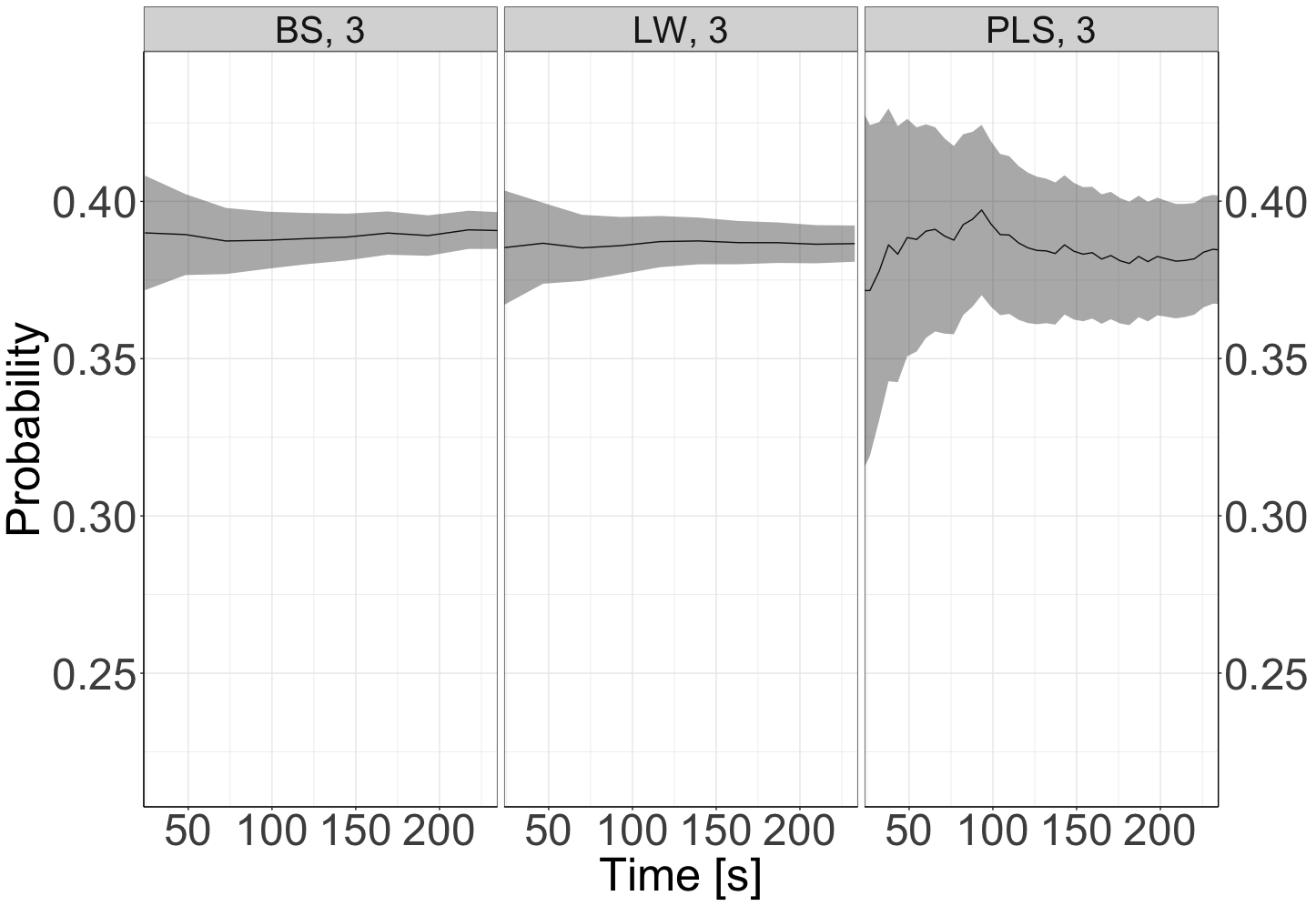}
\end{minipage}
\caption{Convergence of goal node probability for all three techniques with one evidence node \emph{(left)} and three evidence nodes \emph{(right)}}
  \label{fig_conv}
\end{figure}
The techniques are also run across graph sizes of 100 nodes to 1500, again with 30 runs per graph size, with 1, 5 and 10 evidence nodes used. The results of these runs are shown in Figure \ref{fig_size} with the points showing the average run time and the ribbon showing the maximum and minimum of the runs. PLS runs slightly worse than the other two techniques for one evidence node but performs very poorly for the other evidence levels so is omitted from the graphs for clarity. The stopping criteria for a run is an error of $\pm0.02$ per node. 

BS performs best with one evidence node, and gives similar results for both 5 and 10 pieces of evidence. LW performs about as well at all evidence levels for graphs below 1000 nodes, but performs best with 10 evidence nodes for graphs larger than 1000.

\begin{figure}[h!]
  \centering
  \includegraphics[width=0.7\linewidth]{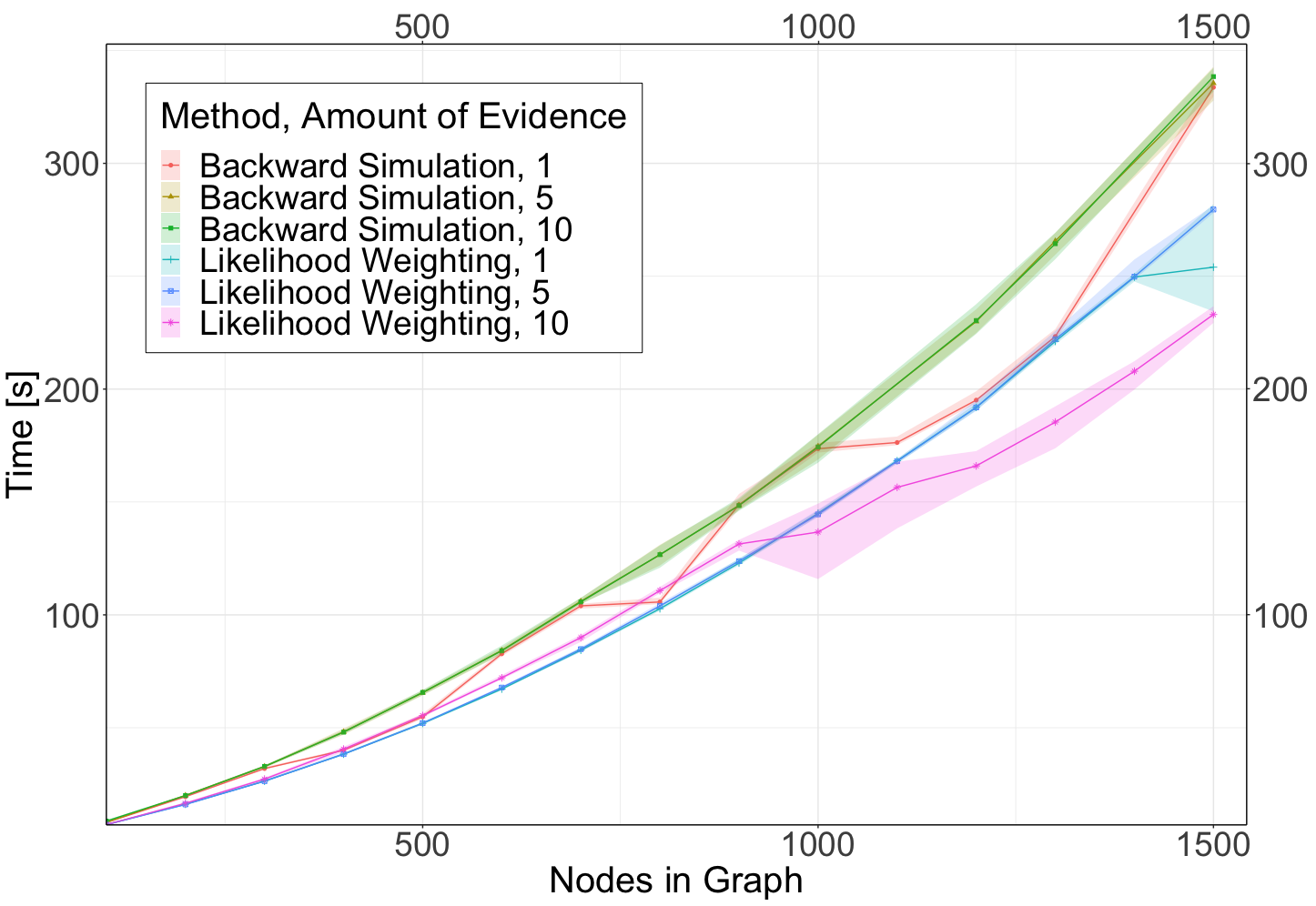}
  \caption{Time required for increasing graph sizes for LW and BS for different amounts of evidence.}
  \label{fig_size}
\end{figure}

Next the stopping criteria is relaxed to an error of $\pm0.04$ per node and the runs are repeated in a similar manner to Figure \ref{fig_size} but up to graphs of 5000 nodes to investigate the scalability of the techniques. These runs are plotted in Figure \ref{fig_size2} where it can be seen that LW out performs BS for all evidence levels at the larger graph sizes.
\begin{figure}[h!]
  \centering
  \includegraphics[width=0.75\linewidth]{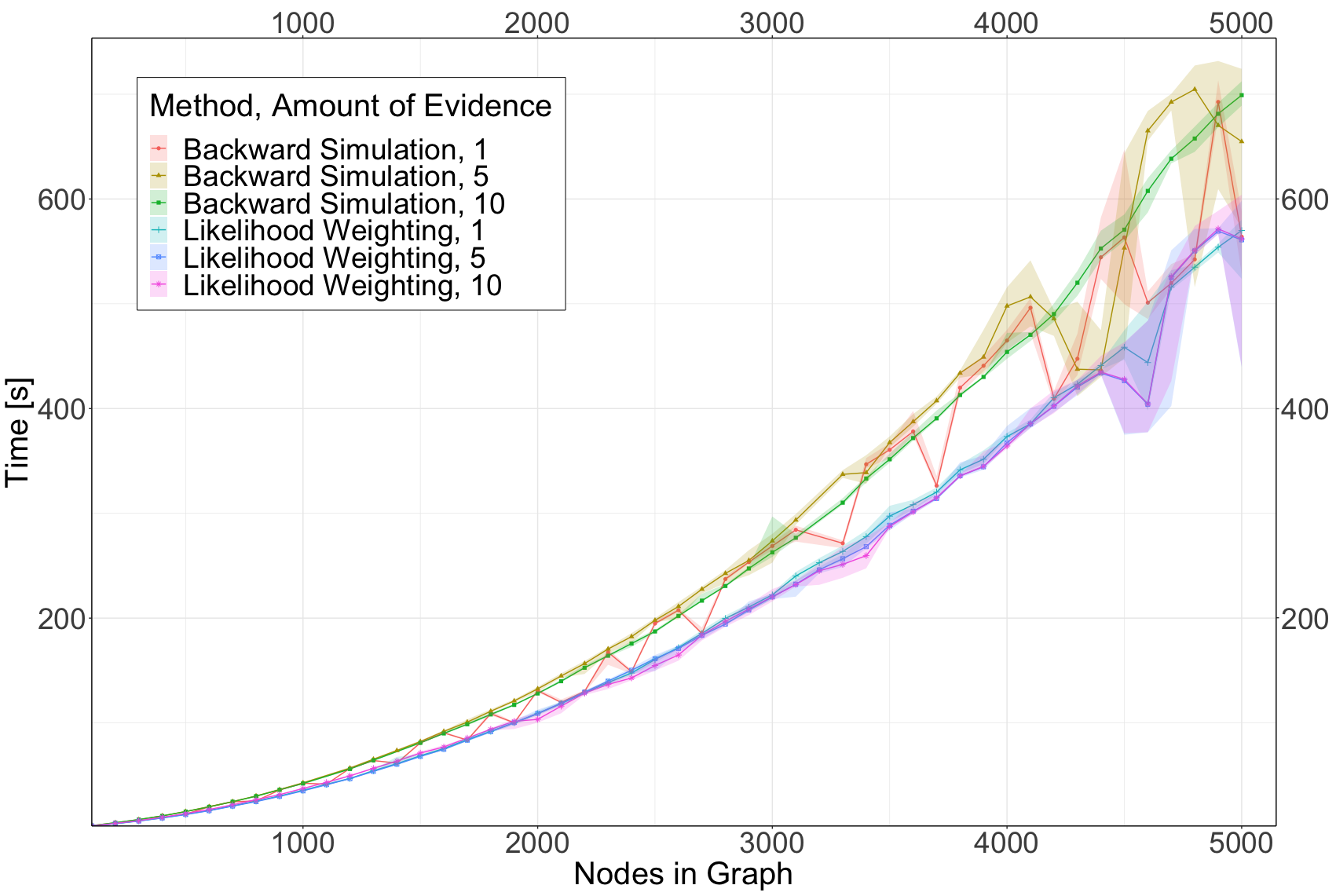}
  \caption{Time required for larger graph sizes for LW and BS for different amounts of evidence.}
  \label{fig_size2}
\end{figure}

Given these results, likelihood weighting is the best technique for belief updating in Bayesian attack graphs as it not only performs slightly better overall across different evidence levels and graph sizes but also is easier to implement than backward simulation. Probabilistic logic sampling should only be used if no evidence is expected most of the time.


\section{Sensitivity Analysis}
\label{sec:sens}
\begin{figure}[h!]
\includegraphics[width=\textwidth]{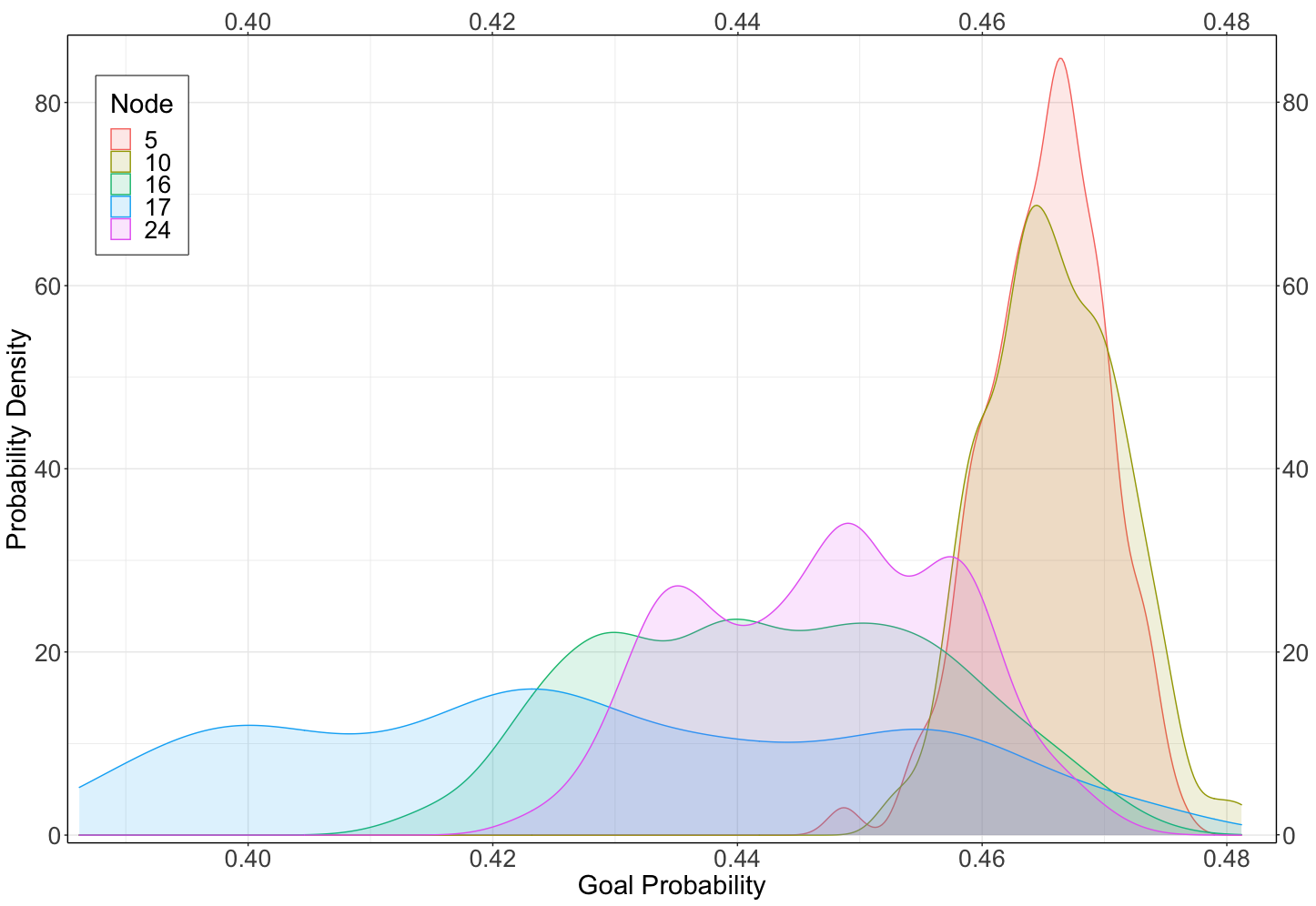}
\caption{Probability density of goal node when a uniform distribution is used for various leaf nodes, demonstrating their sensitivity.} \label{fig_sensdens}
\end{figure}

The prior probabilities for the vulnerabilities on the LEAF nodes can be generated via different methods of varying complexity.
For example Doynikova and Kotenko \cite{Doynikova2017} use various parts of the CVSS vector and Cheng et al. \cite{Cheng2017} model the relationships of parts of the metrics to give them different weights and improve the accuracy of the probabilities. All these techniques however draw from the data available for a vulnerability which is often incomplete and quickly becomes outdated. Sensitivity analysis is important for the overall analysis of BAGs as it considers the impact of the original assignment of the probability.


To evaluate the sensitivity of the graph to the LEAF nodes, each node can be assigned a uniform probability distribution in turn rather than a single probability. A distribution can be generated for one or several goal nodes in the network with respect to each LEAF node; this is done by sampling from the LEAF nodes' uniform distribution, then generating a sample of the entire network as before. The change in the probability density of the access probability of the goal node in the network from Figure \ref{fig:egbag} is shown in Figure \ref{fig_sensdens}. The wider the distribution the more sensitive the goal node is to the probability applied at the LEAF node, if the LEAF node probability does not affect the goal node at all the probability density would be entirely concentrated at the goal probability value that is calculated when there are only single values for all LEAF nodes.

As such, the network is more sensitive to changes at node 17 or 16, whereas nodes 5 and 10 do not have much of an effect on the goal probability as shown by their narrow probability densities. This type of sensitivity analysis has been performed by others, as discussed in Section \ref{sec:related}, however, in what follows we propose an alternative technique that requires much less computation and gives more usable results.

\begin{theorem}
Using the stochastic techniques for incorporating evidence in the graph, one can also discover the sensitivities of LEAF nodes on the graph with much less computational effort than the method of applying a uniform distribution at the LEAF nodes. Instead of generating simulations for a distribution of probabilities, a sensitivity value can be calculated by performing simulations when a LEAF node is given evidence as 'y' and as 'n'. As such only two simulations need to be generated.
\end{theorem}
\begin{proof}
The probability of the goal node, $P(A)$, can be calculated with respect to a LEAF node $B$
\begin{equation}
P(A) = P(A|B=0)P(B=0) + P(A|B=1)P(B=1) 
\end{equation}
which is equivalent to 
\begin{equation}
P(A) = P(A|B=0)(1-P(B=1)) + P(A|B=1)P(B=1).
\end{equation}
We then calculate sensitivity as
\begin{equation}
Sensitivity = P(A|B=1) - P(A|B=0) 
\end{equation}
\end{proof}

The sensitivities calculated in this manner are shown in Table \ref{tab_sens}, and using this sensitivity value allows quick evaluation of the importance of each node without the extensive computation or the required analysis of the probability distribution that is necessary to generate and interpret Figure \ref{fig_sensdens}. The information remains the same, however, with node 17 being the most important followed by 16 then 24, while nodes 5 and 10 have very little impact on the goal node probability.
\begin{table}[h!]
\renewcommand{\arraystretch}{1.3}
\caption{Sensitivities calculated using 'on/off' evidence.}
\label{tab_sens}
\centering
\begin{tabular}{c|c }
\hline
\bfseries Leaf Node & \bfseries Sensitivity \\
\hline\hline
17 & 0.7780 \\
\hline
16 & 0.4388 \\
\hline
24 & 0.3526 \\
\hline
5 & 0.0225 \\
\hline
10 & 0.0081 \\
\hline
\end{tabular}
\end{table}

\section{Related Work}
\label{sec:related}
One example of stochastic simulation techniques for attack graphs is by Noel and Jajodia \cite{Noel2010measuring}. They use PLS to compare different security fixes for a network. However this is performed by hand and as such it cannot be generally applied. Their use case compares several security controls that could be added to the network. This is achieved by examining the resulting distributions estimated when the changes are applied to the graph, in a manner similar to that shown in Figure~\ref{fig_sensdens} as a sensitivity analysis. As discussed in Section~\ref{sec:sens}, this requires more computation and also requires analysis of the resulting distributions. Baiardi and Sgandurra use Monte Carlo simulations in their Haruspex tool \cite{baiardi2013assessing}. This tool is a fully featured program that uses attack graphs and threat agents to model security. It is an application for this type of graph, incorporating many different elements, but does not analyse different methods for simulation.

Muñoz-González et al. present an exact method for inference in BAGs using the junction tree algorithm \cite{Munoz-Gonzalez2016}. This method is attractive due to its exact nature, but unfortunately is very limited in its application due to how it scales. This is caused by the requirement for tables to be generated based on the cliques created to start the calculations, and for large graphs these tables can become extremely large. It is better to have a trade-off in the accuracy of the method to reduce the space required, to allow scalability for the large graphs that are expected from enterprise networks. 
They go on to present an approximate technique in \cite{Munoz-Gonzalez2017} using loopy belief propagation. The results of this scale well, linearly with respect to the number of nodes, while achieving a reasonable level of accuracy. The drawback to using this method, unlike stochastic simulation, is that there is no guarantee of convergence to the correct value. 

\section{Conclusion}
\label{sec:conclusion}

In this paper we have presented and compared three techniques that can be generally applied to inference of any Bayesian attack graph. We make the recommendation that for most purposes the likelihood weighting process is a good choice to analyse an attack graph when any amount of evidence is presented, in a timely fashion. We also demonstrate a test of sensitivity for the graph that can be very quickly calculated and does not require any complex analysis of distributions or prior sampling of node distributions. This can be used both as remediation for the high uncertainty in LEAF node prior probabilities, as well as an easy prioritisation of vulnerabilities in light of their importance to a series of goal nodes.
%
%
%
\bibliographystyle{splncs04}
\bibliography{./references}

\begin{thebibliography}{10}
\providecommand{\url}[1]{\texttt{#1}}
\providecommand{\urlprefix}{URL }
\providecommand{\doi}[1]{https://doi.org/#1}

\bibitem{Aguessy2016}
Aguessy, F., Bettan, O., Blanc, G., Conan, V., Debar, H.: {B}ayesian attack
  model for dynamic risk assessment. CoRR  \textbf{abs/1606.09042} (2016)

\bibitem{baiardi2013assessing}
Baiardi, F., Sgandurra, D.: Assessing {ICT} risk through a {M}onte {C}arlo
  method. Environment Systems and Decisions  \textbf{33}(4),  486--499 (2013)

\bibitem{Cheng2017}
Cheng, P., Wang, L., Jajodia, S., Singhal, A.: Refining {CVSS}-Based Network
  Security Metrics by Examining the Base Scores, pp. 25--52. Springer
  International Publishing (2017)

\bibitem{Dantu2004}
Dantu, R., Loper, K., Kolan, P.: Risk management using behavior based attack
  graphs. In: International Conference on Information Technology: Coding and
  Computing. vol.~1, pp. 445--449 (April 2004)

\bibitem{Doynikova2017}
Doynikova, E., Kotenko, I.: Enhancement of probabilistic attack graphs for
  accurate cyber security monitoring. In: IEEE SmartWorld, Ubiquitous
  Intelligence Computing, Advanced Trusted Computed, Scalable Computing
  Communications, Cloud Big Data Computing, Internet of People and Smart City
  Innovation. pp.~1--6 (Aug 2017)

\bibitem{firstorg}
FIRST: Common vulnerability scoring system v3.1: Specification document (2019),
  \url{https://www.first.org/cvss/v3.1/specification-document}

\bibitem{Frigault2017}
Frigault, M., Wang, L., Jajodia, S., Singhal, A.: Measuring the Overall Network
  Security by Combining CVSS Scores Based on Attack Graphs and {B}ayesian
  Networks, pp. 1--23. Springer International Publishing, Cham (2017)

\bibitem{Fung1994backward}
Fung, R., Del~Favero, B.: Backward simulation in {B}ayesian networks. In:
  Uncertainty Proceedings 1994, pp. 227--234. Elsevier (1994)

\bibitem{openvas}
Greenbone: Open vulnerability assesment scanner (2006),
  \url{https://www.https://www.openvas.org/}

\bibitem{homer2013aggregating}
Homer, J., Zhang, S., Ou, X., Schmidt, D., Du, Y., Rajagopalan, S.R., Singhal,
  A.: Aggregating vulnerability metrics in enterprise networks using attack
  graphs. Journal of Computer Security  \textbf{21}(4),  561--597 (2013)

\bibitem{Huangfu2017}
Huangfu, Y., Zhou, L., Yang, C.: Routing the cyber-attack path with the
  {B}ayesian network deducing approach. In: 2017 International Conference on
  Cyber-Enabled Distributed Computing and Knowledge Discovery (CyberC). pp.
  5--10 (Oct 2017)

\bibitem{Ingols2009}
Ingols, K., Chu, M., Lippmann, R., Webster, S., Boyer, S.: Modeling modern
  network attacks and countermeasures using attack graphs. In: Annual Computer
  Security Applications Conference. pp. 117--126 (Dec 2009)

\bibitem{jajodia2009topological}
Jajodia, S., Noel, S.: Topological vulnerability analysis: A powerful new
  approach for network attack prevention, detection, and response. In:
  Algorithms, architectures and information systems security, pp. 285--305.
  World Scientific (2009)

\bibitem{Keramati2014}
Keramati, M., Keramati, M.: Novel security metrics for ranking vulnerabilities
  in computer networks. In: 7th International Symposium on Telecommunications
  (IST). pp. 883--888 (Sept 2014)

\bibitem{matthews2020cyclic}
Matthews, I., Mace, J., Soudjani, S., van Moorsel, A.: Cyclic {B}ayesian attack
  graphs: A systematic computational approach (2020)

\bibitem{Munoz-Gonzalez2016}
Mu{\~{n}}oz{-}Gonz{\'{a}}lez, L., Sgandurra, D., Barrere, M., Lupu, E.: Exact
  inference techniques for the analysis of {B}ayesian attack graphs. IEEE
  Transactions on Dependable and Secure Computing  (2016)

\bibitem{Munoz-Gonzalez2017}
Mu{\~{n}}oz{-}Gonz{\'{a}}lez, L., Sgandurra, D., Paudice, A., Lupu, E.C.:
  Efficient attack graph analysis through approximate inference. CoRR
  \textbf{abs/1606.07025} (2017)

\bibitem{Nielsen2009bayesian}
Nielsen, T.D., Jensen, F.V.: Bayesian networks and decision graphs. Springer
  Science \& Business Media (2009)

\bibitem{Noel2010measuring}
Noel, S., Jajodia, S., Wang, L., Singhal, A.: Measuring security risk of
  networks using attack graphs. International Journal of Next-Generation
  Computing  \textbf{1}(1),  135--147 (2010)

\bibitem{Ou2006:SAA:1180405.1180446}
Ou, X., Boyer, W.F., McQueen, M.A.: A scalable approach to attack graph
  generation. In: Proceedings of the 13th ACM Conference on Computer and
  Communications Security. pp. 336--345. CCS '06, ACM, New York, NY, USA (2006)

\bibitem{Ou2005:MLN:1251398.1251406}
Ou, X., Govindavajhala, S., Appel, A.W.: {MulVAL}: A logic-based network
  security analyzer. In: Proceedings of the 14th Conference on USENIX Security
  Symposium. SSYM'05, vol.~14, pp.~8--8. USENIX Association, Berkeley, CA, USA
  (2005)

\bibitem{Ou2011}
Ou, X., Singhal, A.: Attack Graph Techniques, pp.~5--8. Springer New York, New
  York, NY (2011)

\bibitem{Poolsappasit2012}
Poolsappasit, N., Dewri, R., Ray, I.: Dynamic security risk management using
  {B}ayesian attack graphs. IEEE Transactions on Dependable and Secure
  Computing  \textbf{9}(1),  61--74 (Jan 2012)

\bibitem{Ramaki2015}
Ramaki, A.A., Khosravi-Farmad, M., Bafghi, A.G.: Real time alert correlation
  and prediction using {B}ayesian networks. In: 12th International Iranian
  Society of Cryptology Conference on Information Security and Cryptology
  (ISCISC). pp. 98--103 (Sept 2015)

\bibitem{sembiring2015network}
Sembiring, J., Ramadhan, M., Gondokaryono, Y.S., Arman, A.A.: Network security
  risk analysis using improved {MulVAL} {B}ayesian attack graphs. International
  Journal on Electrical Engineering and Informatics  \textbf{7}(4), ~735 (2015)

\bibitem{Swiler2001}
{Swiler}, L.P., {Phillips}, C., {Ellis}, D., {Chakerian}, S.: Computer-attack
  graph generation tool. In: Proceedings DARPA Information Survivability
  Conference and Exposition II. DISCEX'01. vol.~2, pp. 307--321 vol.2 (June
  2001)

\bibitem{nessus}
Tenable: Nessus vulnerability scanner (1998),
  \url{https://www.tenable.com/products/nessus/nessus-professional/}

\end{thebibliography}

\clearpage
\appendix
\section{Full Example}
\label{app:eg}
\begin{figure*}[ht]
    \centering
	\includegraphics[width=\linewidth]{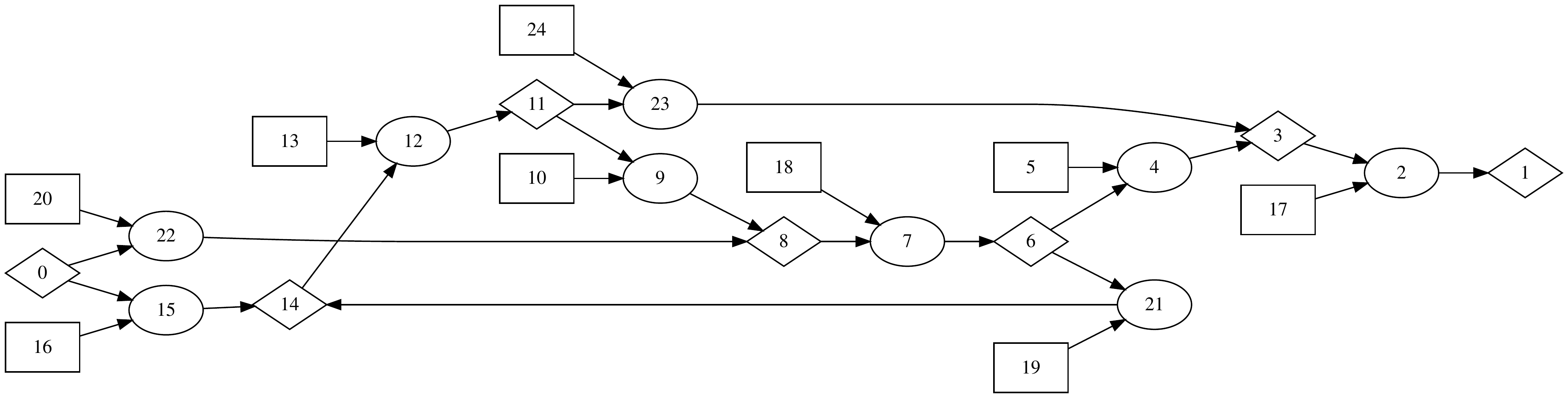}
	\caption{The BAG of the running example including leaf nodes.}
	\label{fullbagbw}
\end{figure*}
The complete attack graph for the running example scenario can be seen in Figure \ref{fullbagbw}, with the labels for the nodes shown in Listing \ref{listing}.
The vulnerabilities in this scenario are described below
\begin{itemize}
    \item CVE-2009-2446\footnote{https://nvd.nist.gov/vuln/detail/CVE-2009-2446} is present on the Database Server. This is a MySQL vulnerability that allows an authenticated user to cause a denial of service as well as execute arbitrary code
    \item CVE-2006-3747\footnote{https://nvd.nist.gov/vuln/detail/CVE-2006-3747} is on the Webserver. This is a vulnerability in the Apache web server that requires network access and is exploited using specially crafted URLs that then allow the attacker to execute arbitrary code
    \item CVE-2009-1918\footnote{https://nvd.nist.gov/vuln/detail/CVE-2009-1918} is a vulnerability on all the Workstations. It affects Internet Explorer, and means that if the user visits a website with malformed elements, a memory corruption is triggered that an attacker can use to execute code
\end{itemize}

\lstset{
  basicstyle=\ttfamily,
  columns=fullflexible,
  keepspaces=true,
}
\begin{lstlisting}[caption={MulVAL labels for Figure \ref{fullbagbw}}, label={listing}]
0, "attackerLocated(internet)"
1, "execCode(dbServer,root)"
2, "RULE 2 (remote exploit of a server 
program)"
3, "netAccess(dbServer,tcp,'3306')"
4, "RULE 5 (multi-hop access)"
5, "hacl(webServer,dbServer,
tcp,'3306')"
6, "execCode(webServer,apache)"
7, "RULE 2"
8, "netAccess(webServer,tcp,'80')"
9, "RULE 5"
10, "hacl(workStation,webServer,tcp,'80'"
11, "execCode(workStation,userAccount)"
12, "RULE 2"
13, "vulExists(workStation,'CVE-2009-1918',
IE,remoteExploit,privEscalation)"
14, "accessMaliciousInput(workStation, user, IE)"
15, "malicious website"
16, "visit of malicious website"
17, "vulExists(dbServer,'CVE-2009-2446',
mySQL,remoteExploit,privEscalation)"
18, "vulExists(webServer,'CVE-2006-3747',
apache,remoteExploit,privEscalation)"
19, "visit of compromised website"
20, "hacl(internet, webServer, tcp, '80')"
21, "compromise of website"
22, "RULE 6 (direct network access)"
23, "RULE 5"
24, "hacl(workStation,dbServer,tcp,'3306')
\end{lstlisting}
An \texttt{hacl} node is a 'host access control list' node that defines which hosts can connect to other hosts, as well as the protocol that is allowed to be used and the port that it is performed through.


\end{document}